%% file: main.tex
\documentclass[12pt]{article}

\usepackage{amsthm}
\usepackage{amsfonts}
\usepackage{graphicx}
\usepackage{algorithm}
\usepackage{algpseudocode}
\usepackage{fullpage}
\usepackage{comment}
\usepackage{complexity} 

\newtheorem{theorem}{Theorem}[section]
\newtheorem{lemma}[theorem]{Lemma}
\newtheorem{corollary}[theorem]{Corollary}

\newtheorem{conjecture}[theorem]{Conjecture}
\newtheorem{problem}[theorem]{Problem}
\newenvironment{definition}[1][Definition]{\begin{trivlist}
\item[\hskip \labelsep {\bfseries #1}]}{\end{trivlist}}

\newcommand{\aesat}{$\forall \exists \mathrm{SAT}$}

\newcommand{\rmSigma}{\mathrm{\Sigma}}
\newcommand{\rmPi}{\mathrm{\Pi}}

\newcommand{\pibv}[1]{PIBV$_{#1}$}
\newcommand{\uibv}[1]{UIBV$_{#1}$}
\newcommand{\tibv}[1]{TIBV$_{#1}$}

\newcommand{\true}{\textsc{T}}

\author{Robert Schweller\thanks{University of Texas Rio Grande Valley, \protect{\texttt{\{robert.schweller,andrew.winslow,timothy.wylie\}@utrgv.edu}}. This research was supported in part by National Science Foundation Grants CCF-1117672 and CCF-1555626.} \and Andrew Winslow\footnotemark[1] \and Tim Wylie\footnotemark[1]}
\title{Verification in Staged Tile Self-Assembly}
\date{}

\begin{document}

\maketitle

\begin{abstract}
We prove the unique assembly and unique shape verification problems, benchmark measures of self-assembly model power, are $\coNP^{\NP}$-hard and contained in \PSPACE{} (and in $\rmPi^\P_{2s}$ for staged systems with $s$ stages).
En route, we prove that unique shape verification problem in the 2HAM is $\coNP^{\NP}$-complete.
\end{abstract}

\input{introduction}

\input{definitions}

\input{2hamReduction}

\input{level1}

\input{level2}

\input{pspace}

\input{open}

\bibliographystyle{abbrv} 
\bibliography{tam}

\end{document}

%% file: introduction.tex
\section{Introduction}
\label{sec:intro}

Here we consider the complexity of two standard problems in tile self-assembly: deciding whether a system uniquely assembles a given assembly or shape.
These so-called \emph{unique assembly} and \emph{unique shape verification} problems are benchmark problems in tile assembly, and have been studied in a variety of models, including the aTAM~\cite{ACGHKMR02,Bryans-2013a}, the $q$-tile model~\cite{AGKS05g}, and the 2HAM~\cite{CDDEPSSW13}.

The unique assembly and unique shape verification problems ask whether a system behaves as expected: does a given system yield a unique given assembly or assemblies of a given unique shape?
The distinct rules by which assemblies form in various tile assembly models yield the potential for such problems to have varying complexity.
For instance, assuming $\P \neq \NP$, the unique \emph{assembly} verification problem is known to be a strictly easier problem in the aTAM than in the 2HAM.

However, several open questions remain.
For instance, such a separation between the aTAM and 2HAM for the unique \emph{shape} verification problem had not been known.
Here we prove such a separation (see Table~\ref{tab:compare}).

\bgroup
\setlength{\tabcolsep}{0.5em}
\def\arraystretch{1.2}

\begin{table}
\centering
\begin{tabular}{|c|c|c|}
\hline
Model & Unique Assembly & Unique Shape \\
\hline
aTAM & \P~\cite{ACGHKMR02} & \coNP{}-complete~\cite{AGKS05g} \\
\hline
2HAM & \coNP{}-complete~\cite{2HOT2publish} & $\coNP^{\NP}$-complete (Sec.~\ref{sec:2HAM}) \\
\hline
Staged & \multicolumn{2}{|c|}{$\coNP^{\NP}$-hard~(Sec.~\ref{sec:level2}), in $\PSPACE{}$~(Sec.~\ref{sec:pspace})} \\
\hline
\end{tabular}
\vspace{1em}
\caption{\label{tab:compare} Known and new results on the unique assembly and unique shape verification problems.}
\end{table}
\egroup

Additionally, a popular generalization of the 2HAM called the \emph{staged tile assembly model}~\cite{DDFIRSS07} has been shown to be capable of extremely efficient assembly across a range of parameters~\cite{CMS2016OSS,DDFIRSS07,demaine2011ODS,DFS2015NGA,Winslow-2015a}.
Does this power come from the increased complexity of verifying that systems assemble intended assemblies and shapes?

We achieve progress on these questions, proving a separation between the 2HAM and staged model for the unique assembly verification problem ($\coNP$-complete versus $\coNP^{\NP}$-hard) utilizing a promising technique that may lead to proving a stronger separation for the unique shape verification problem ($\coNP^{\NP}$-complete versus a conjectured \PSPACE-complete).

The $\coNP^{\NP}$-hardness results are also interesting as the first, to our knowledge, verification problems in irreversible tile assembly that are decidable but not contained in \NP{} or \coNP{}.


%% file: definitions.tex
\section{The Staged Assembly Model}
\label{sec:model}

\textbf{Tiles.}
A \emph{tile} is a non-rotatable unit square with each edge labeled with a \emph{glue} from a set $\Sigma$.
Each pair of glues $g_1, g_2 \in \Sigma$ has a non-negative integer \emph{strength}, denoted ${\rm str}(g_1, g_2)$.
Every set $\Sigma$ contains a special \emph{null glue} whose strength with every other glue is 0.
If the glue strengths do not obey ${\rm str}(g_1,g_2) = 0$ for all $g_1 \neq g_2$, then the glues are \emph{flexible}.
Unless otherwise stated, we assume that glues are not flexible.

\textbf{Configurations, assemblies, and shapes.}
A \emph{configuration} is a partial function $A : \mathbb{Z}^2 \rightarrow T$ for some set of tiles $T$, i.e., an arrangement of tiles on a square grid.
For a configuration $A$ and vector $\vec{u} = \langle u_x, u_y \rangle \in \mathbb{Z}^2$, $A + \vec{u}$ denotes the configuration $f \circ A$, where $f(x, y) = (x + u_x, y + u_y)$.
For two configurations $A$ and $B$, $B$ is a \emph{translation} of $A$, written $B \simeq A$, provided that $B = A + \vec{u}$ for some vector $\vec{u}$.
For a configuration $A$, the \emph{assembly} of $A$ is the set $\tilde{A} = \{ B : B \simeq A \}$.
An assembly $\tilde{A}$ is a \emph{subassembly} of an assembly $\tilde{B}$, denoted $\tilde{A} \sqsubseteq \tilde{B}$, provided that there exists an $A\in \tilde{A}$ and $B\in \tilde{B}$ such that $A \subseteq B$.
The \emph{shape} of an assembly $\tilde{A}$ is $\{ {\rm dom}(A) : A \in \tilde{A}\}$ where dom() is the domain of a configuration.
A shape $S'$ is a \emph{scaled} version of shape $S$ provided that for some $k \in \mathbb{N}$ and $D \in S$, $\bigcup_{(x, y) \in D} \bigcup_{(i, j) \in \{0, 1, \dots, k-1\}^2} (kx + i, ky + j) \in S'$.

\textbf{Bond graphs and stability.}
For a configuration~$A$, define the \emph{bond graph}~$G_A$ to be the weighted grid graph in which each element of~${\rm dom}(A)$ is a vertex, and the weight of the edge between a pair of tiles is equal to the strength of the coincident glue pair.
A configuration is \emph{$\tau$-stable} for $\tau \in \mathbb{N}$ if every edge cut of $G_A$ has strength at least $\tau$, and is \emph{$\tau$-unstable} otherwise.
Similarly, an assembly is \emph{$\tau$-stable} provided the configurations it contains are $\tau$-stable.
Assemblies $\tilde{A}$ and $\tilde{B}$ are \emph{$\tau$-combinable} into an assembly $\tilde{C}$ provided there exist $A \in \tilde{A}$, $B \in \tilde{B}$, and $C \in \tilde{C}$ such that $A \bigcup B = C$, $\rm dom(A) \bigcap \rm dom(B) = \emptyset$, and $\tilde{C}$ is $\tau$-stable.

\textbf{Two-handed assembly and bins.}
We define the assembly process via bins.
A bin is an ordered tuple $(S,\tau)$ where $S$ is a set of \emph{initial} assemblies and $\tau \in \mathbb{N}$ is the \emph{temperature}.
In this work, $\tau$ is always equal to $2$ for upper bounds, and at most some constant for lower bounds.
For a bin $(S, \tau)$, the set of \emph{produced} assemblies $P'_{(S,\tau)}$ is defined recursively as follows:
\vspace*{-.2cm}
\begin{enumerate}
\item $S \subseteq P'_{(S,\tau)}$.
\item If $A,B \in P'_{(S,\tau)}$ are $\tau$-combinable into $C$, then $C \in P'_{(S,\tau)}$.
\end{enumerate}
\vspace*{-.2cm}
A produced assembly is \emph{terminal} provided it is not $\tau$-combinable with any other producible assembly, and the set of all terminal assemblies of a bin $(S,\tau)$ is denoted $P_{(S,\tau)}$.
That is, $P'_{(S,\tau)}$ represents the set of all possible assemblies that can assemble from the initial set $S$, whereas $P_{(S,\tau)}$ represents only the set of assemblies that cannot grow any further.

The assemblies in $P_{(S,\tau)}$ are \emph{uniquely produced} iff for each $x \in P'_{(S, \tau)}$ there exists a corresponding $y \in P_{(S,\tau)}$ such that $x \sqsubseteq y$.
Unique production implies that every producible assembly can be repeatedly combined with others to form an assembly in $P_{(S,\tau)}$.

\textbf{Staged assembly systems.}
An \emph{$r$-stage $b$-bin mix graph} $M$ is an acyclic $r$-partite digraph consisting of $rb$ vertices $m_{i,j}$ for $1 \leq i \leq r$ and $1\leq j \leq b$, and edges of the form $(m_{i,j}, m_{i+1,j'})$ for some $i, j, j'$.
A \emph{staged assembly system} is a 3-tuple $\langle M_{r,b}, \{T_1, T_2, \dots, T_b\}, \tau \rangle$ where $M_{r,b}$ is an $r$-stage $b$-bin mix graph, $T_i$ is a set of tile types, and $\tau \in \mathbb{N}$ is the temperature.
Given a staged assembly system, for each $1\leq i \leq r$, $1\leq j \leq b$, a corresponding bin $(R_{i,j}, \tau)$ is defined as follows:
\vspace*{-.2cm}
\begin{enumerate}
\item $R_{1,j}= T_j$ (this is a bin in the first stage);
\item For $i\geq 2$,
$\displaystyle R_{i,j}= \Big(\bigcup_{k:\ (m_{i-1,k},m_{i,j})\in M_{r,b}} P_{(R_{(i-1,k)},\tau_{i-1,k})}\Big)$.
\end{enumerate}
\vspace*{-.2cm}
Thus, bins in stage 1 are tile sets $T_j$, and each bin in any subsequent stage receives an initial set of assemblies consisting of the terminally produced assemblies from a subset of the bins in the previous stage as dictated by the edges of the mix graph.\footnote{The original staged model~\cite{DDFIRSS07} only considered $O(1)$ distinct tile types, and thus for simplicity allowed tiles to be added at any stage (since $\mathcal{O}(1)$ extra bins could hold the individual tile types to mix at any stage). Because systems here may have super-constant tile complexity, we restrict tiles to only be added at the initial stage.}
The \emph{output} of a staged system is the union of the set of terminal assemblies of the bins in the final stage.\footnote{This is a slight modification of the original staged model~\cite{DDFIRSS07} in that there is no requirement of a final stage with a single output bin. This may be a slightly more capable model, and so it is considered here. However, all results in this paper apply to both variants of the model.}
The output of a staged system is \emph{uniquely produced} provided each bin in the staged system uniquely produces its terminal assemblies.

%% file: 2hamReduction.tex
\section{The 2HAM Unique Shape Verification Problem is $\coNP^{\NP}$-complete}
\label{sec:2HAM}

This section serves as a warm-up for the format and techniques used in later sections.
We begin by proving the 2HAM USV problem is in $\coNP^{\NP}$ by providing a (non-deterministic) algorithm for the problem that can be executed on such a machine.
This is followed by a reduction from a \SAT-like problem complete for $\coNP^{\NP}$ (\aesat{}).

\begin{definition}[2HAM unique shape verification (2HAM USV) problem]
Given a 2HAM system $\Gamma$ and shape $S$, does every terminal assembly of $\Gamma$ have shape $S$?
\end{definition}

\begin{theorem}\label{thm:2HAMUSVHard}
The 2HAM USV problem (for $\tau=2$ systems) is $\coNP^{\NP}$-hard.
\end{theorem}

\begin{definition}[\aesat]
Given a 3-\SAT{} formula $\phi(x_1, x_2, \dots, x_k, x_{k+1}, \dots, x_n)$, is it true that for every assignment of $x_1, x_2, \dots, x_k$, there exists an assignment of $x_{k+1}, x_{k+2}, \dots, x_n$ such that $\phi(x_1, x_2, \dots, x_n)$ evaluates to \true?
\end{definition}

The \aesat{} problem was shown to be $\coNP^{\NP}$-complete by Stockmeyer~\cite{Stockmeyer-1976a} (see~\cite{Schaefer-2002a} for further discussion).

\begin{proof}
The reduction is from \aesat{}.
Roughly speaking, the system output by the reduction behaves as follows.
First, a distinct assembly encoding each possible assignment of the variables of the \aesat{} instance is assembled.
Further growth ``tags'' each assembly as either a \emph{true} or \emph{false} assembly, based upon the truth value of the input 3-\SAT{} formula $\phi$ for the variable assignment encoded by the assembly.

False assemblies further grow into a slightly larger target shape $S$.
A separate set of \emph{test} assemblies are created, one for each variable assignment of the variables $x_1, \dots x_k$.
Each test assembly attaches to any true assembly with the same assignment of these variables to form an assembly with shape $S$ - the same shape as false assemblies.

Terminal assemblies then consist of false assemblies and true-test assemblies with shape $S$, and possibly test assemblies.
A test assembly is terminal if and only if there is no true assembly for it to attach to, i.e. the assignment of variables $x_1, \dots, x_k$ has no corresponding assignment of the variables $x_{k+1}, \dots, x_n$ such that $\phi(x_1, \dots, x_n) = \true$.

\begin{figure}[t]
\includegraphics[width=\textwidth]{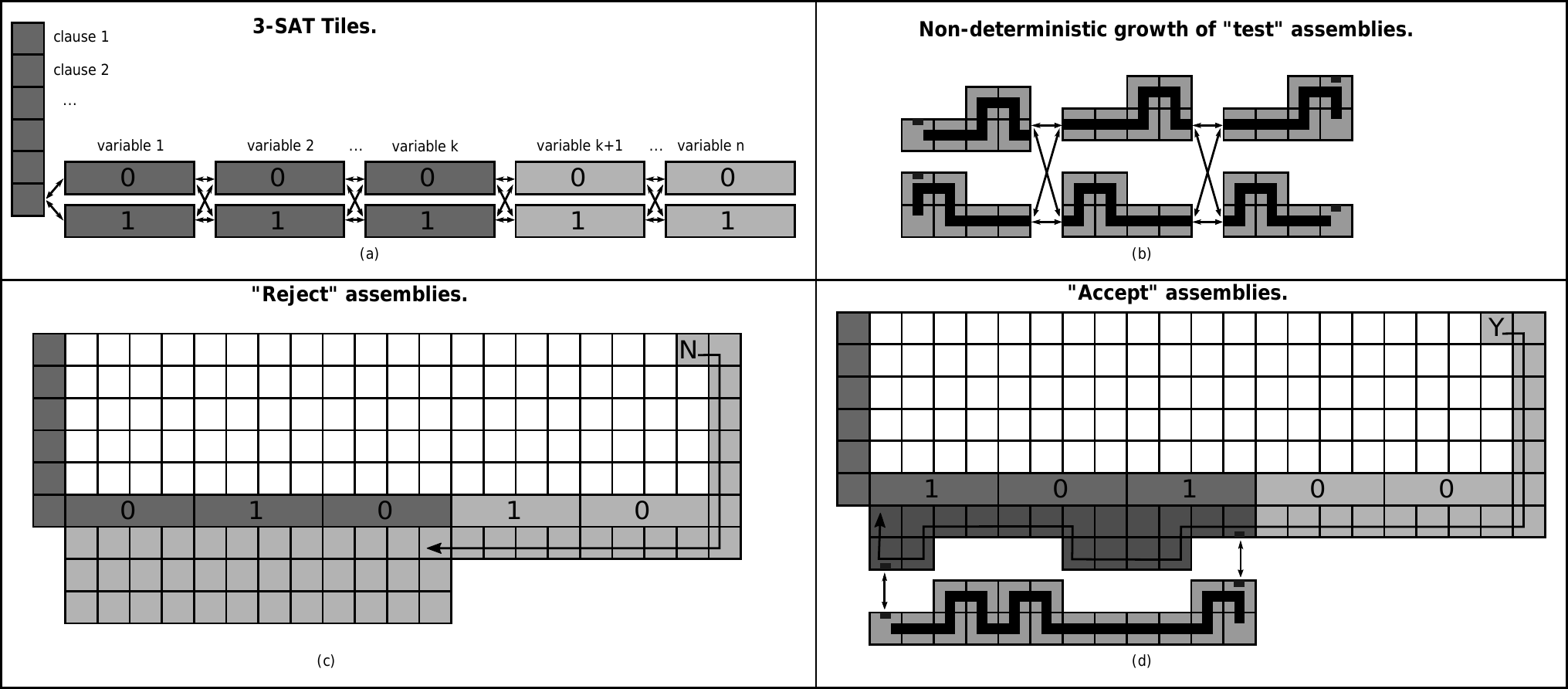}
\caption{Steps of the 2HAM USV $\coNP^{\NP}$-hardness reduction.}
\label{fig:2hamRedux}
\end{figure}

\textbf{\SAT{} assemblies.}
Consider a given input formula $C$ and input value $k$ for the \aesat{} problem.
From this input we design a corresponding 2HAM system $\Gamma = (T,2)$ and shape $S$ such that the terminal assemblies of $\Gamma$ share a common shape $S$ if and only if the \aesat{} instance is ``true'', i.e. each assignment of the variables $x_1$ through $x_k$ can be combined with some assignment of the variables $x_{k+1}$ through $x_m$ such that the 3-\SAT{} instance is satisfied.

The system has temperature~2, and the tile set $T$ of the system output by the reduction is sketched in Figure~\ref{fig:2hamRedux}.
The first subset of tiles is a minor modification of the commonly used 3-\SAT{} solving system from~\cite{Lagoudakis992ddna}.

For each variable $x_i$, the system has two tile subsets.
These collections assemble into $1\times 4$ assemblies with exposed north and south glues representing the values ``0'' and ``1'', respectively, encoding the assignment of a specific variable to true or false.
These $1 \times 4$ assemblies further assemble into $1 \times 4n$ assemblies encoding complete assignments of the variables $x_1$ to $x_n$.
The non-deterministic assembly process of 2HAM implies that such an assembly for every possible variable assignment will be assembled.

An additional column is attached to this bar of height equal to $m$, the number of clauses in the formula $C$ (Figure~\ref{fig:2hamRedux}).
An additional set of tiles are added that evaluate the 3-\SAT{} formula $\phi$ based upon the variable assignments encoded by the initial $1 \times 4n$ assembly following the approach of~\cite{Lagoudakis992ddna}.
These tiles place a tile in the upper right corner of the resulting assembly with exposed glue labeled ``T'' or ``F'', indicating the truth value of $\phi$ based upon the variable assignments.

The resulting assemblies are categorized as \emph{true} and \emph{false} assemblies.
Additional tiles are added so that every false assembly further grows, extending the left $4k$ columns (corresponding to the variables $x_1$ to $x_k$) southward by~3 rows, and the remaining right $4(n-k)$ columns southward by~1 row (Figure~\ref{fig:2hamRedux}(c)).
The resulting shape is the shape $S$ output by the reduction, i.e. the only shape assembled by the system if the solution to the \aesat{} instance is ``true''.

\textbf{Test assemblies.}
Additional tiles are also added so that true assemblies also grow southward, but extending the left $4k$ columns by various amounts based upon each variable assignment.
The result is a sequence of geometric ``bumps and dents'' that encode the truth values of these variables.

A set of \emph{test} assemblies with complementary geometry for each possible assignment of variables $x_1$ through $x_k$ are assembled (Figure~\ref{fig:2hamRedux}(b)).
Test assemblies use two strength-1 glues that cooperatively attach to any true assembly with a matching assignment of variables $x_1$ through $x_k$ (Figure~\ref{fig:2hamRedux}(d)).
The assembly formed by a test assembly attaching to a true assembly has shape $S$: the same shape as a false assembly.

\textbf{Terminal assemblies.}
If the solution to the \aesat{} instance is ``false'', there is some truth assignment for variables $x_1 \ldots x_k$ with no corresponding assignment of the variables $x_{k+1} \dots x_n$ such that $\phi(x_1, \dots, x_n)$ is ``true''.
Thus, the test assembly with this assignment of variables $x_1, \dots, x_k$ has no compatible true assembly to attach to - and this test assembly is a terminal assembly of $\Gamma$ with shape not equal to $S$.

On the other hand, if the solution to the \aesat{} instance is ``true'', every test assembly attaches to a true assembly and thus every terminal assembly (true-test assemblies and false assemblies) has shape $S$.
\end{proof}

\begin{theorem}
The 2HAM USV problem is in $\coNP^{\NP}$.
\end{theorem}

\begin{proof}
The solution to an instance $(\Gamma, S)$ of the 2HAM USV problem is ``true'' if and only if:
\begin{enumerate}
\item Every producible assembly of $\Gamma$ has size at most $|S|$.
\item Every assembly of size at most $|S|$ and without shape $S$ is not a terminal assembly.
\end{enumerate}
Algorithm~\ref{alg:2HAM-USV} solves the 2HAM USV problem by verifying each of these conditions, using an \NP{} subroutine to verify the second condition.
The algorithm is executed by a \coNP{} machine, implying that ``false'' is returned if any of the non-deterministic branches return ``false'', and otherwise returns ``true''.

\begin{algorithm}[h!]
\caption{\label{alg:2HAM-USV} A $\coNP^{\NP}$ algorithm for the 2HAM USV problem}
\begin{algorithmic}[1]
\State Non-deterministically select a $\tau$-stable assembly $A$ with $|S| < |A| \leq 2|S|$.
\If {$A$ is producible} \Comment{In \P{} by Theorem 3.2 of~\cite{Doty-2014a}}
\State \textbf{return} false.
\EndIf
\State Non-deterministically select a $\tau$-stable assembly $B$ with $|B| \leq |S|$ and shape not equal to $S$.
\If {not $\mathcal{F}(\Gamma, B, |S|)$} \Comment{Algorithm~\ref{alg:2HAM-USV-sub}}
\State \textbf{return} false.
\EndIf
\State \textbf{return} true.
\end{algorithmic}
\end{algorithm}

\begin{algorithm}[h!]
\caption{\label{alg:2HAM-USV-sub} An \NP{} algorithm subroutine of Algorithm~\ref{alg:2HAM-USV}}
\begin{algorithmic}[1]
\Procedure{$\mathcal{F}$}{$\Gamma, B, n$} \Comment{Returns whether $B$ is \emph{not} terminal.}
\State Non-deterministically select a $\tau$-stable assembly $C$ with $|C| \leq n$.
\If {$C$ cannot attach to $B$ at temperature $\tau$}
\State \textbf{return} false.
\EndIf
\If {$C$ is a producible assembly of $\Gamma$} \Comment{In \P{} by Theorem 3.2 of~\cite{Doty-2014a}}
\State \textbf{return} false.
\EndIf
\State \textbf{return} true.
\EndProcedure
\end{algorithmic}
\end{algorithm}
\end{proof}

%% file: level1.tex
\section{Staged Unique Assembly Verification is $\coNP$-hard}
\label{sec:level1}

\begin{figure}[t]
 	\centering \includegraphics[width=1.0\textwidth]{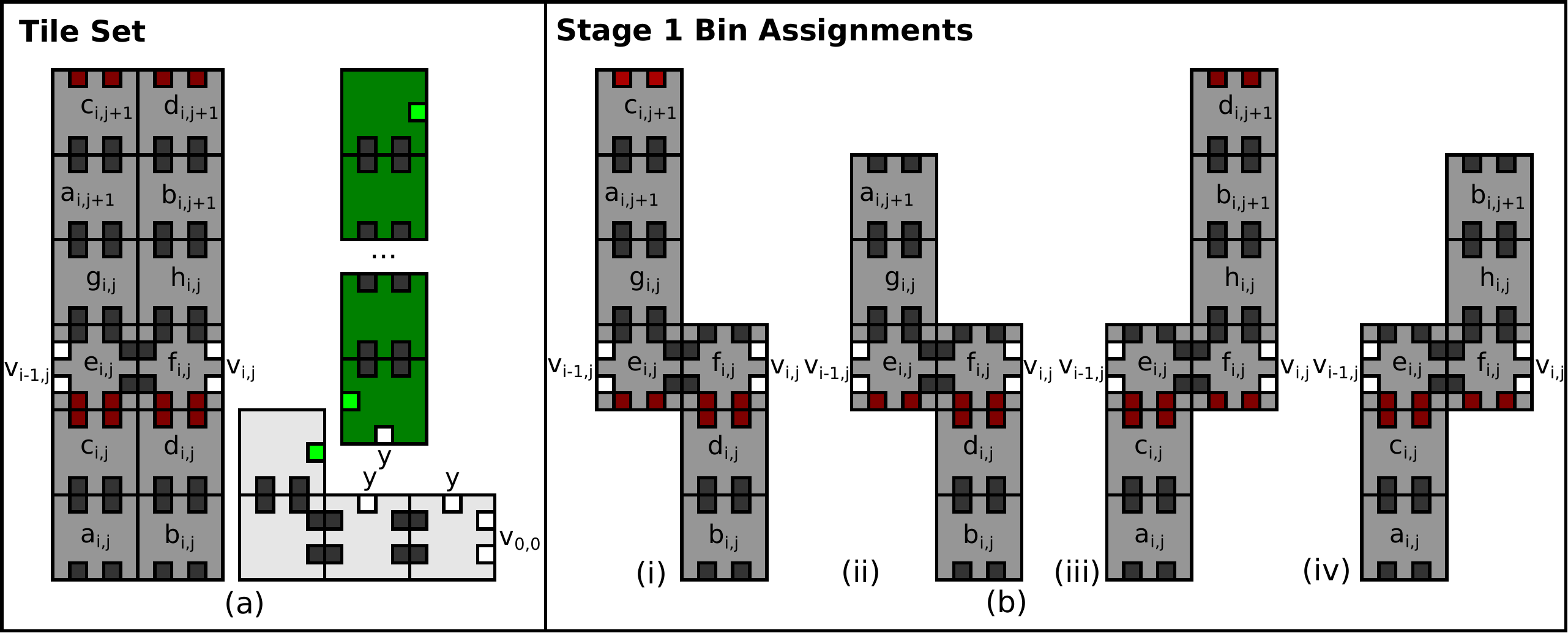}
 	\caption{(a) The tile set used in the staged \coNP-hardness reduction.  (b) The subsets of tiles included in separated initial bins within the first stage of the system.}
 	\label{fig:uavNPHardTileset}
\end{figure}

\begin{figure}[t]
 	\centering \includegraphics[width=1.0\textwidth]{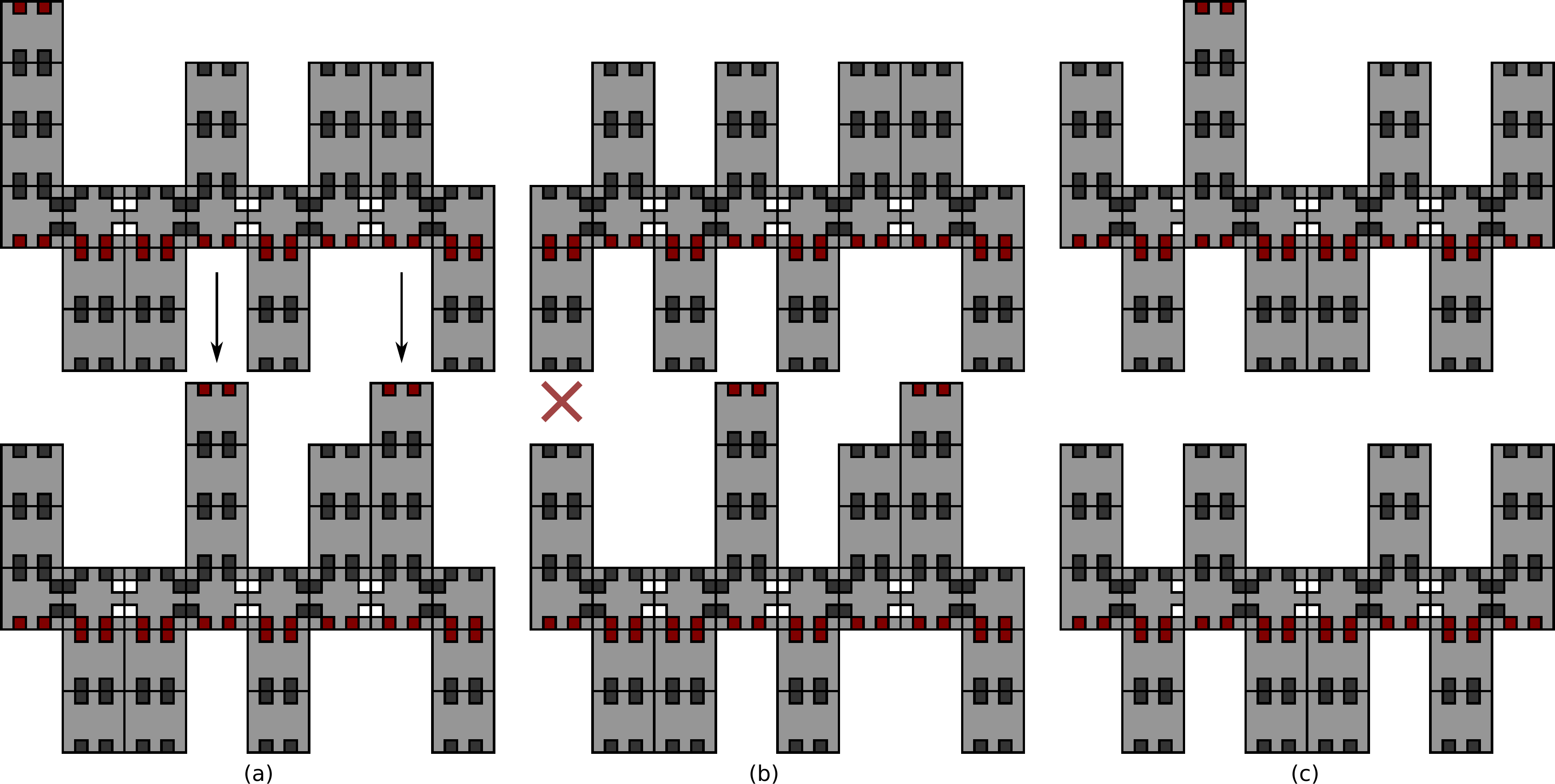}
 	\caption{In stage~2, rows non-deterministically form encoding each of the $2^n$ possible variable assignments.  In stage~3 the rows are combined allowing for geometrically compatible, sequential rows with exposed red glue to attach. (a) Combinable rows.  (b) Geometrically incompatible rows.  (c) Rows with no glues for attachment.}
 	\label{fig:uavRows}
\end{figure}

\begin{figure}[t]
 	\centering \includegraphics[width=1.0\textwidth]{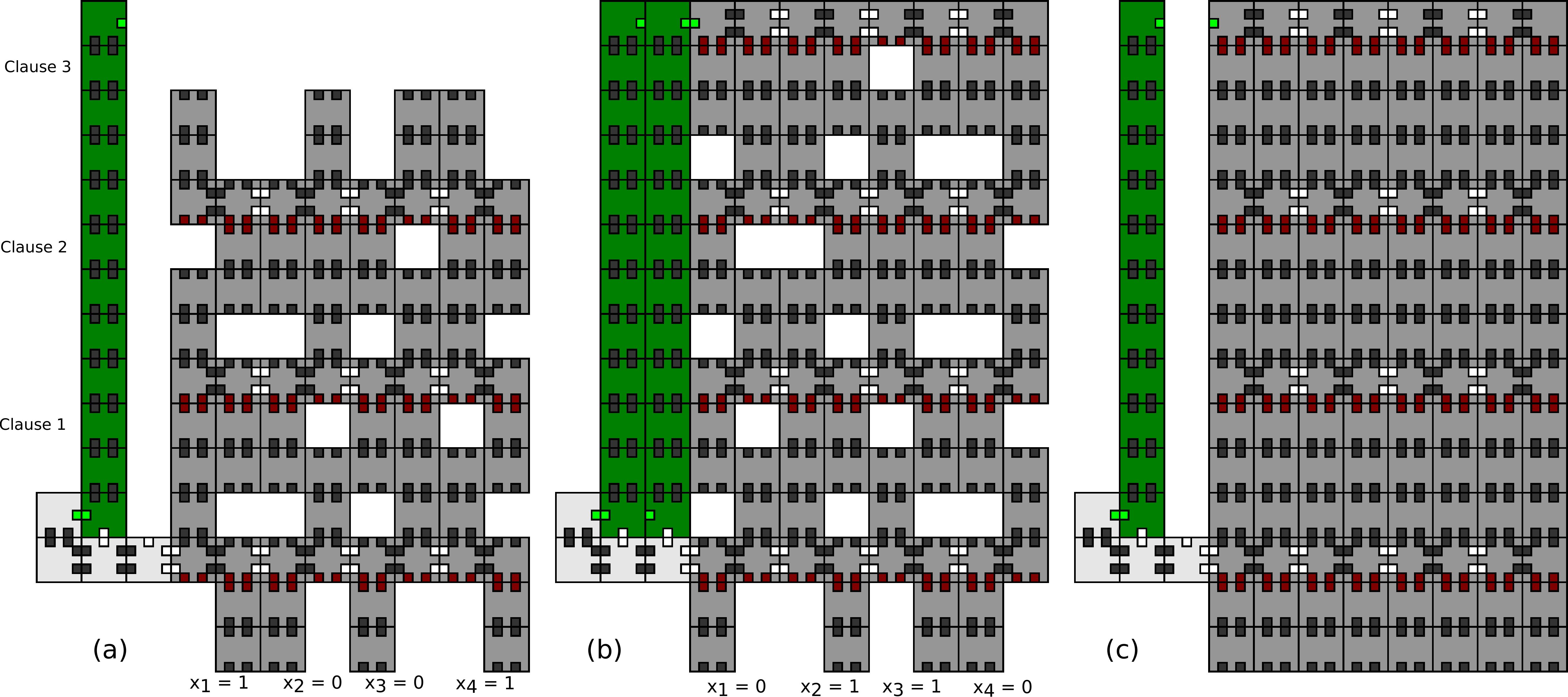}
 	\caption{(a) Non-satisfying variable assignments will not be able to grow from row 0 to row $m$.  (b) Assemblies encoding satisfying variable assignments will allow for complete assemblies with all rows, allowing for a green assembly to attach. (c) The target assembly $A$ given as output of the reduction.}
 	\label{fig:uavFinalAssembly}
\end{figure}

\begin{definition}[Staged unique assembly verification (Staged UAV) problem]
Given a staged system $\Gamma$ and an assembly $A$, does $\Gamma$ uniquely assemble $A$?
\end{definition}

\begin{theorem}
\label{thm:stagedUAVcoNP}
The staged UAV problem (for 4-stage systems at $\tau=2$) is \coNP-hard.
\end{theorem}

\begin{proof}
The reduction is from 3-\SAT{}, outputting a staged system $\Gamma$ and assembly $A$ such that the 3-\SAT{} instance is satisfiable if and only if $A$ is \emph{not} the unique terminal assembly of $\Gamma$. 
We reduce from 3-SAT:  Given a 3-SAT formula $\phi$, we design a staged assembly system and an assembly $A$ such that $\phi$ is \emph{not} satisfied if and only if $A$ is uniquely assembled by $\Gamma$.  

\textbf{The tileset.}  The tiles used in our construction are shown in Figure~\ref{fig:uavNPHardTileset}(a).  
In particular, for each variable $x_i \in \{x_1, x_2, \dots, x_n\}$ and clause $c_j \in \{c_1, c_2, \dots, c_m\}$ in $\phi$, there is a block of tiles labeled $a_{i,j},b_{i,j},c_{i,j},d_{i,j},e_{i,j},f_{i,j},g_{i,j}$.
The set of tile types for each block is denoted $\mathrm{block}_{i, j}$.

The strength-2 ($\tau=2$) glues connecting adjacent tiles are unique with respect to adjacent tiles, and are unlabelled in the figures for clarity.  
Note that for each block $(i, j)$, the top four tiles of the block occupy the same locations as the bottom four tiles of block $(i,j+1)$.  
Finally, the tileset includes a length $4m$ chain of \emph{green} tiles, with each green tile sharing a strength-2 glue with its neighbors, along with four light-grey tiles which together attach to the green assembly.

\textbf{Stage 1: variable assignments.}
The specific formula $\phi$ is encoded within the output staged system via the initial choice of tiles placed into a $O(1)$-sized collection of stage-1 bins.
For each variable $x_i$ and clause $c_j$ combination, we select two subsets of the $\mathrm{block}_{i,j}$ tileset.  
The first subset encodes a variable choice of ``false'' for $x_i$.
The tile sets in Figure~\ref{fig:uavNPHardTileset}(b)(i) and (iv) are used if $x_i$ satisfies (and $\overline{x_i}$ does not satisfy) clause $c_j$, respectively.
Similarly, the tile sets in Figure~\ref{fig:uavNPHardTileset}(b)(ii-iii) are used if $x_i$ does not (and $\overline{x_i}$ does satisfy) clause $c_j$.

Beyond utilizing two types of $\mathrm{block}_{i,j}$ tile sets, tile sets are further distinguished between odd and even values of $i$ and $j$.
In total,~16 distinct bins (satisfied or not, negated or not, odd or even $i$, odd or even $j$) are used.  

We include the grey and green tiles of Figure~\ref{fig:uavNPHardTileset}(a) separately in two additional bins.  
An additional four bins are used in the construction to maintain a set of single copies of all tiles used within the system.  
Separating these tile subsets into four bins ensures that the tiles do no interact (until mixed with other assemblies at a later stage).

\textbf{Stage 2: assembling rows.}
In stage 2 we combine all $\mathrm{block}_{i,j}$ assemblies for even $j$ into one bin, and all $\mathrm{block}_{i,j}$ assemblies for odd $j$ into a second bin.  
Within each bin and for each value $j$, rows encoding each possible variable assignment assemble non-determistically via attaching $0-\mathrm{block}_{i,j}$ and $1-\mathrm{block}_{i,j}$ assemblies for each $i \in \{1, 2, \dots, n\}$.
We refer to these assemblies as $\mathrm{row}_j$ assemblies.
There are $2^n$ such assemblies for each $j$ - one per variable assignment.
Example $\mathrm{row}_j$ assemblies are shown in Figure~\ref{fig:uavRows}.

\textbf{Stage 3: combining rows with shared assignments and satisfied clauses.}
Stage~3 is where the \emph{real} action happens.  
All $\mathrm{row}_j$ assemblies are combined, along with the green and grey assemblies of Figure~\ref{fig:uavNPHardTileset}.  

Consider the possible assembly of a $\mathrm{row}_j$ and a $\mathrm{row}_{j+1}$ assembly.  
If the two respective rows encode distinct variable assignments, geometric incompatibility prohibits any possible connection (Figure~\ref{fig:uavRows}(b)).  
If the rows encode the same truth assignment, then the rows may attach if any of the $\mathrm{row}_{j}$ variable pieces expose the extended tip via the red $\tau=2$ strength glues (Figure~\ref{fig:uavRows}(a)).  
Such an attachment indicates that the variable assignment of both rows satisfies $c_j$.  
If the variable assignment encoding does not satisfy $c_j$, no extended tip exists and the rows cannot attach (Figure~\ref{fig:uavRows}(c)).  

A satisfying assignment of $\phi$ corresponds to $m$ rows attaching to form a complete ``satisfying'' assembly (Figure~\ref{fig:uavFinalAssembly}(b)).  
The green assembly attaches cooperatively to such assemblies using the $\mathrm{row}_m$ assembly glue and a glue from the grey tiles, which attach uniquely to $\mathrm{row}_0$.  
The attachment of a green assembly verifies that all rows are present and the variable assignment satisfies $\phi$.

A second copy of the green assembly attaches to any assembly containing $\mathrm{row}_0$, regardless of whether all rows are present or not (Figure~\ref{fig:uavFinalAssembly}(a)).  
In a separate bin, the green assembly tiles and grey assemblies are combined, yielding a combined grey-green product (for mixing in stage~4).

\textbf{Stage 4: merging assignments.}
In stage~4, the set of all $\mathrm{block}_{i,j}$ individual tiles are added to the assemblies constructed in stage~3 as well as the the grey-green assembly produced in the previous stage.  
Note that the green assembly is \emph{not} an input assembly to this mixing.

Since all $\mathrm{block}_{i,j}$ assemblies are included, each terminal assembly from stage~3 may grow into the unique terminal assembly shown in Figure~\ref{fig:uavFinalAssembly}(c) with one exception: assemblies from stage~3 encoding satisfying variable assignments.
These assemblies have one additional copy of the green bar assembly attached. 
Therefore, the assembly of Figure~\ref{fig:uavFinalAssembly}(c) is uniquely assembled if an only if no such satisfying assembly exists.
\end{proof}

%% file: level2.tex
\section{Staged Unique Assembly Verification is $\coNP^{\NP}$-hard}
\label{sec:level2}

\begin{figure}[t]
 	\centering \includegraphics[width=\textwidth]{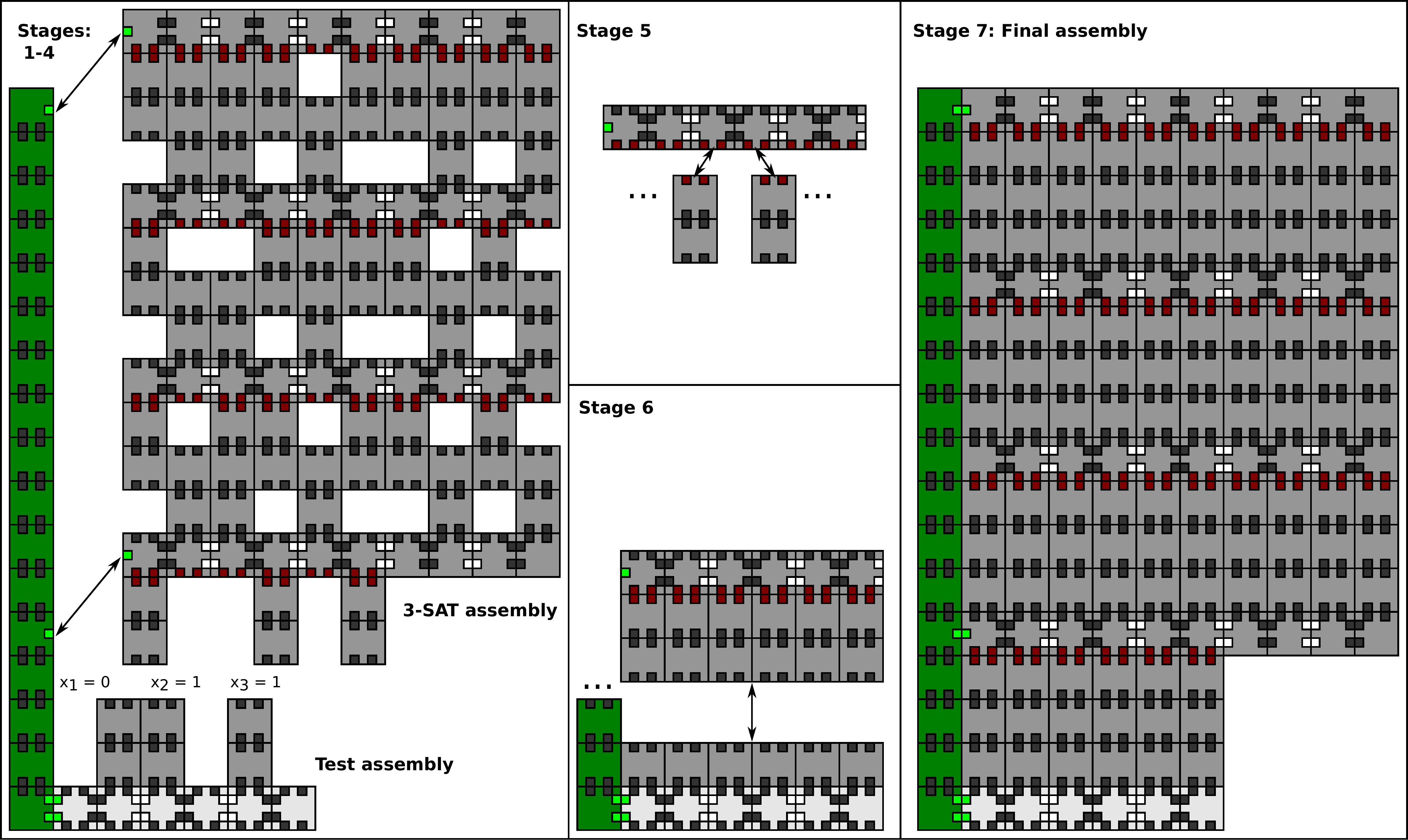}
 	\caption{The assemblies at respective stages for the $\coNP^{\NP}$-hardness reduction for the staged UAV problem.}
 	\label{fig:uavLevel2}
\end{figure}

\begin{theorem}\label{thm:stagedUAVExtraHard}
The staged UAV problem (for $\tau=2$ 7-stage systems) is $\coNP^{\NP}$-hard.
\end{theorem}

\begin{proof}
We reduce from \aesat{} by combining ideas from the reductions of Theorem~\ref{thm:2HAMUSVHard} and~\ref{thm:stagedUAVcoNP}.

\textbf{Stages 1-3: the \SAT{} assemblies.}  
The first~3 stages follows those of the reduction in Theorem~\ref{thm:stagedUAVcoNP} but without the inclusion of the green assembly and light grey tiles.  
The result is a collection of assemblies encoding satisfying variable assignments with all $m$ rows, as well as partial assemblies of less than $m$ rows encoding non-satisfying assignments.
For clarity, the bottom half of the $j=0$ blocks for values $i>k$ are removed, exposing the ``geometric teeth" only for the first $k$ variables.

\textbf{Stages 1-3: the test assemblies.}  
Additionally, in a separate set of bins, we non-deterministically generate a set of \emph{test} assemblies.  
The test assemblies are similar to row assemblies and generated in a similar fashion.  
An example test assembly is shown in Figure~\ref{fig:uavLevel2} (Stages 1-4).  
A test assembly for each of the $2^k$ possible truth assignments of $x_1, x_2, \dots, x_k$ is grown, and a green bar assembly is attached to the side of each test assembly.

\textbf{Stage 4: the magic happens.}
The \SAT{} assemblies and test assemblies are combined in a bin.
Test assemblies attach to \SAT{} assesmblies encoding satisfying variable assignments by utilizing cooperative bonding based on the two strength-1 green glues on the green assembly.  
\SAT assemblies encoding non-satisfying assignments must each lack the topmost or bottommost row, and therefore cannot attach to a test assembly.

Due to the geometric interlocking teeth from the test assembly and the bottom of \SAT{} assemblies, test assemblies may only attach to \SAT{} assemblies that encode the same variable assignment (of variables $x_1, x_2, \dots, x_k$).
Stages 1-4 of Figure~\ref{fig:uavLevel2} show an example test assembly and a attaching \SAT{} assembly.

Note that if there exists a truth assignment for $x_1, x_2, \dots, x_k$ with no satisfying assignment of the remaining variables $x_{k+1}, x_{k+2}, \dots, x_n$, then the corresponding test assembly does not attach to \emph{any} \SAT{} assembly and is a terminal assembly of this bin.
On the other had, if every assignment of the variables $x_1, x_2, \dots, x_k$ has at least one satisfying assignment of the remaining variables, i.e. the solution \aesat{} instance is ``true'', then there are no terminal test assemblies of this bin

\textbf{Stage 5: tagging non-satisfying assignments.} In Stage~5, we add preassembled duples which attach to the bottom of any assembly containing row 0 and encodes a non-satisfying variable assignment.
This attachment ensures that in subsequent stages, these assemblies will be geometrically incompatible with any remaining test assemblies from Stage~4.

It is possible that some duples have no non-satisfying \SAT{} assembly to attach to.
As a solution, an additional height-1 assembly of the row-0 assembly that ``absorbs'' each duple is added at this stage.
The subsequent stages enable these, as well as all other \SAT{} assemblies, to grow into a single common (potentially) unique assembly.

\textbf{Stage 6: attaching test assemblies.}
The result of Stage~5 is mixed with an assembly consisting of:
\begin{itemize}
\item The light-grey bar of the test assemblies.
\item A second complete layer of dark grey tiles.
\item The green bar.
\end{itemize}
This assembly attaches to any non-satisfying \SAT{} assembly that includes row~0, ensuring that all assemblies containing row~0 now have a version of the test assembly attached (Stage~6 in Figure~\ref{fig:uavLevel2}).

\textbf{Stage 7: merging.}  In the final stage, every individual tile of the target assembly (seen in Stage~7 of Figure~\ref{fig:uavLevel2}) is added to the result of Stage~6, with the exception of the green tiles and the tiles in rows~1 through~5 of the \SAT{} assemblies.

These tiles complete each \SAT{} assembly in the assembly in Figure~\ref{fig:uavLevel2} (Stage~7).  
Morever, the height-1 assembly used to absorb duples from Stage~5 grows into the assembly from Figure~\ref{fig:uavLevel2} (Stage~7).  
However, because of the lack of tiles from rows~1 through~5, any leftover test assembly from Stage~4 remains terminal.

Thus the target assembly is the unique terminal assembly of the system if and only if the solution to the \aesat{} instance is ``yes''.
\end{proof}

Observe that every staged system output by the reduction has the property that if it does not have a unique terminal assembly, then it also does not have a unique terminal shape.
Thus the same reduction suffices to prove that the staged USV problem is $\coNP^{\NP}$-hard.

\begin{corollary}
The staged USV problem is $\coNP^{\NP}$-hard.
\end{corollary}

%% file: pspace.tex
\section{Staged \PSPACE{} containment}
\label{sec:pspace}

Here we prove that the staged UAV and USV problems are in \PSPACE{}.
Parameterized versions of the results are also obtained; these prove that both problems restricted to systems with any \emph{fixed} number of stages lie in the polynomial hierarchy.
Both results are obtained via upper bounds on the complexities of the following three problems: 

\begin{definition}[Stage-$s$ producible-in-bin verification (\pibv{s}) problem] 
Given a staged system $\Gamma$, a bin $b$ in stage $s$ of $\Gamma$, an assembly $A$, and an integer $n$:
\begin{enumerate}
\item is $A$ a producible assembly of $b$? 
\item and does every producible assembly of every bin in stage $s-1$ of $\Gamma$ have size at most $n$? 
\end{enumerate}
\end{definition}

\begin{definition}[Stage-$s$ undersized-in-bin verification (\uibv{s}) problem]
Given a staged system $\Gamma$, a bin $b$ in stage $s$ of $\Gamma$, and an integer $n$: 
\begin{enumerate}
\item and does every producible assembly of $b$ have size at most $n$? 
\item and does every producible assembly of every bin in stage $s-1$ of $\Gamma$ have size at most $n$? 
\end{enumerate}
\end{definition}

\begin{definition}[Stage-$s$ terminal-in-bin verification (\tibv{s}) problem] 
Given a staged system $\Gamma$, a bin $b$ in stage $s$ of $\Gamma$, an assembly $A$, and an integer $n$:
\begin{enumerate}
\item is $A$ a terminal assembly of $b$?
\item and does every producible assembly of $b$ have size at most $n$? 
\item and does every producible assembly of every bin in stage $s-1$ of $\Gamma$ have size at most $n$? 
\end{enumerate}
\end{definition}

The statements and proofs of the following results use terminology related to the polynomial hierarchy.
For an introduction to the polynomial hierarchy, see Stockmeyer~\cite{Stockmeyer-1976a}.
As a reminder, $\rmSigma^\P_{i+1} = \NP^{\rmSigma^{\P}_i}$, $\rmPi^\P_{i+1} = \coNP^{\rmSigma^{\P}_i}$, and $\rmSigma^\P_0 = \rmPi^\P_0 = \P$.

\begin{lemma}
\label{lemma:hier}
For all $s \in \mathbb{N}$:
\begin{itemize}
\item The \pibv{s} problem is in $\rmSigma^\P_{2s-2}$.
\item The \uibv{s} and \tibv{s} problems are in $\rmPi^\P_{2s-1}$.
\end{itemize}
\end{lemma}

Due to space limitations, the proof of this lemma is omitted.

\begin{proof}
The proof is by induction on $s$.
We begin by proving that \pibv{1} $\in \rmSigma^\P_{2s-2} = \P$ and \uibv{1}, \tibv{1} $\in \rmPi^\P_{2s-1} = \coNP$ (the base case).
Then we provide recursive algorithms of the correct complexity for \pibv{s}, \uibv{s}, and \tibv{s}, assuming that such algorithms exist for \pibv{s-1}, \uibv{s-1}, and \tibv{s-1} (the inductive step).

\textbf{Algorithms for the \pibv{1}, \uibv{1}, and \tibv{1} problems.}
All three problems contain, as a subproblem, ``does every producible assembly of every bin in stage $s-1$ of $\Gamma$ have size at most $n$?''.
The answer to this is trivially yes - so only the complexity of the other subproblems needs consideration.

Theorem 3.2 of Doty~\cite{Doty-2014a} states that there exists a polynomial-time algorithm for \pibv{1}.
The \uibv{1} problem can be solved by a \coNP{} machine via non-deterministically selecting an assembly of size in $(n, 2n]$ consisting of tile types input into bin $b$ and returning ``no'' if the assembly is producible (the machine returns ``no'' if any non-deterministic branch returns ``no'').
The \tibv{1} problem can be solved by a \coNP{} machine by (1) returning ``no'' if $A$ is not producible, (2) returning ``no'' if a second assembly (non-deterministically selected) is producible and attaches to $A$, (3) returning ``yes'' otherwise. 

\textbf{An algorithm for the \pibv{s} problem.}
We now assume from now on that there exist algorithms $\mathcal{P}_{s-1}$, $\mathcal{U}_{s-1}$, and $\mathcal{T}_{s-1}$ for the \pibv{s-1}, \uibv{s-1}, and \tibv{s-1} problems in $\rmSigma^\P_{2s-4}$, $\rmPi^\P_{2s-3}$, and $\rmPi^\P_{2s-3}$, respectively, by the inductive hypothesis. 
\begin{algorithm}[H]
\caption{A $\rmSigma^\P_{2s-2}$ algorithm for the \pibv{s} problem}
\begin{algorithmic}[1]
\Procedure{$\mathcal{P}_s$}{$\Gamma, b, A, n$} \Comment{Bin $b$ is in stage $s$ of $\Gamma$}
\If {\textbf{not} $A$ is $\tau$-stable} \Comment{In \P{} via min-cut}
\State \textbf{return} no.
\EndIf
\State $I \gets \{A\}$ 
\While {non-deterministically choosing to continue \textbf{and} $|I| < |A|$}
\State Decompose an assembly $B$ in $I$ into two stable subassemblies $B_1$, $B_2$.
\State $I = (I - B) \cup \{B_1, B_2\}$ \Comment{Replace $B$ with $B_1$ and $B_2$}
\EndWhile 
\State Non-deterministically assign a bin $b_{B_i}$ in stage $s-1$ to each $B_i \in I$. 
\ForAll{$B_i \in I$}
	\If{\textbf{not} $\mathcal{T}_{s-1}(\Gamma, b_{B_i}, B_i, n)$ } \Comment{Function call is in $\rmPi^\P_{2s-3}$} 
	\State \textbf{return} no.
	\EndIf
\EndFor
\ForAll{bins $b'$ in stage $s-1$} \Comment{Subproblem~2}
	\If{\textbf{not} $\mathcal{U}_{s-1}(\Gamma, b', n)$ } \Comment{Function call is in $\rmPi^\P_{2s-3}$} 
	\State \textbf{return} no.
	\EndIf
\EndFor
\State \textbf{return} yes.
\EndProcedure
\end{algorithmic}
\end{algorithm}
The algorithm runs as an \NP{} machine (making calls to other machines).
Lines 5-10 non-deterministically compute an assembly process for $A$ in bin $b$, and lines 8-12 check that such a process begins with terminal assemblies of (specific) input bins.
Lines 13-18 simply check that the condition of subproblem~2 is satisfied. 

The complexity of the algorithm is \NP{} with polynomially many calls to algorithms in $\rmPi^\P_{2s-3}$.
That is, $\NP^{\rmPi^\P_{2s-3}} = \NP^{\rmSigma^\P_{2s-3}}  = \rmSigma^\P_{2s-2}$.

\textbf{An algorithm for the \uibv{s} problem.}
Since we have already proved that there exists a $\rmSigma^\P_{2s-2}$ algorithm $\mathcal{P}_s$, we assume this as well. 
\begin{algorithm}[H]
\caption{A $\rmPi^\P_{2s-1}$ algorithm for the \uibv{s} problem}
\begin{algorithmic}[1]
\Procedure{$\mathcal{U}_s$}{$\Gamma, b, n$} \Comment{Bin $b$ is in stage $s$ of $\Gamma$}
\State Non-deterministically select an assembly $A$ with $n < |A| \leq 2n$.
	\If { $\mathcal{P}_s(\Gamma, b, A, n)$ } \Comment{Function call is in $\rmSigma^\P_{2s-2}$}
	\State \textbf{return} no.
	\EndIf
\ForAll{bins $b'$ in stage $s-1$}
	\If { $\mathcal{P}_s(\Gamma, b', A, n)$ } \Comment{Function call is in $\rmSigma^\P_{2s-4}$}
	\State \textbf{return} no.
	\EndIf
\EndFor
\State \textbf{return} yes.
\EndProcedure
\end{algorithmic}
\end{algorithm}

The algorithm runs as a \coNP{} machine, returning ``no'' unless every non-deterministic branch returns ``yes''.
Lines 2-5 solve subproblem~1, while lines 6-10 address subproblem~2. 

The complexity of the algorithm is then \coNP{} with two calls to algorithms in $\rmSigma^\P_{2s-2}$.
That is, $\coNP^{\rmSigma^\P_{2s-2}} = \rmPi^\P_{2s-1}$.

\textbf{An algorithm for the \tibv{s} problem.}
Since we have already proved that there exists a $\rmPi^\P_{2s-1}$ algorithm $\mathcal{U}_s$, we assume this as well. 

\begin{algorithm}[H]
\caption{An $\rmPi^\P_{2s-1}$ algorithm for the \tibv{s} problem}
\begin{algorithmic}[1]
\Procedure{$\mathcal{T}_s$}{$\Gamma, b, A, n$} \Comment{Bin $b$ is in stage $s$ of $\Gamma$}
\If { \textbf{not} $\mathcal{P}_s(\Gamma, b, A, n)$ } \Comment{Function call in $\rmSigma^\P_{2s-2}$}
\State \textbf{return} no.
\EndIf
\State Non-deterministically select an assembly $B$ with $|B| \leq n$.
\If { $\mathcal{P}_s(\Gamma, b, B, n)$ \textbf{and} $A$ and $B$ can attach at temperature $\tau$}
\State \textbf{return} no.
\EndIf
\If { \textbf{not} $\mathcal{U}_s(\Gamma, b, n)$ } \Comment{Subproblems 2 and 3}
\State \textbf{return} no.
\EndIf
\State \textbf{return} yes.
\EndProcedure
\end{algorithmic}
\end{algorithm}

The algorithm runs as a \coNP{} machine, returning ``no'' unless every non-deterministic branch returns ``yes''.
Lines 2-8 verify that $A$ is a terminal assembly of bin $b$ (subproblem~1): $A$ is not a terminal assembly if and only if (1) $A$ is not producible (lines 2-4), or (2) another producible assembly $B$ can attach to $A$ (lines 5-8). 

The complexity of the algorithm needs a slightly careful analysis.
Lines 2-8 can be seen as a \coNP{} algorithm with two calls to algorithms in $\rmSigma^\P_{2s-2}$, i.e. a $\coNP^{\rmSigma^\P_{2s-2}} = \rmPi^\P_{2s-1}$ algorithm.
Then the entire algorithm is a \P{} algorithm with a call to a $\rmPi^\P_{2s-1}$ algorithm (lines 2-8) and another call to a $\rmPi^\P_{2s-1}$ algorithm (line 9).
That is, a $\P^{\rmPi^\P_{2s-1}} = \rmPi^\P_{2s-1}$ algorithm.

\textbf{A remark on the reoccurring subproblem.}
All three problems have the subproblem ``does every producible assembly of every bin in stage $s-1$ of $\Gamma$ have size at most $n$?''
Removing this subproblem from the \tibv{s} problem makes the problem undecidable, since arbitrarily large assemblies (carrying out unbounded computation) may attach to $A$.
Seen from another perspective, line~5 of $\mathcal{T}_s$ is only correct because we may assume that any attaching assembly $B$ has size at most $n$.
The \pibv{s} and \uibv{s} problems are also similarly undecidable when the subproblem is removed.

In a system with a unique terminal assembly/shape, no producible assembly of any bin has size exceeding that the unique terminal assembly/shape. 
Thus adding such a subproblem does not change the answer to staged UAV/USV problem instances (a ``no'' with the added subproblem implies a ``no'' without it as well).
\end{proof}

With this algorithmic machinery in place, we move to the first main result:

\begin{definition}[Stage-$s$ unique assembly verification (Stage-$s$ UAV) problem]
Given a staged system $\Gamma$ with $s$ stages and an assembly $A$, is $A$ the unique terminal assembly of $\Gamma$?
\end{definition}

\begin{theorem}
\label{thm:stage-s-UAV}
The stage-$s$ UAV problem is in $\rmPi^\P_{2s}$.
\end{theorem}

\begin{proof}
We give an algorithm for the stage-$s$ UAV problem.
The stage-$s$ UAV problem may be restated as:
\begin{enumerate}
\item is every assembly $B$ with $|B| \leq |A|$ and $B \neq A$ not a terminal assembly of any bin in stage $s$?
\item and does every producible assembly of every bin in stage $s-1$ of $\Gamma$ have size at most $|A|$? 
\end{enumerate}

In the algorithm below, $\mathcal{T}_s$ and $\mathcal{U}_s$ are algorithms for the \tibv{s} and \uibv{s} problems, respectively. 

\begin{algorithm}[H]
\label{alg:stage-s-UAV}
\caption{A $\rmPi^\P_{2s}$ algorithm for the stage-$s$ UAV problem}
\begin{algorithmic}[1]
\Procedure{$\mathcal{UAV}_s$}{$\Gamma, A$} \Comment{$\Gamma$ has $s$ stages.}
\State Non-deterministically select an assembly $B$ with $|B| \leq n$ and $A \neq B$.
\ForAll {bins $b$ in stage $s$ of $\Gamma$}
	\If {$\mathcal{T}_s(\Gamma, b, B)$} \Comment{Function call is in $\rmPi^\P_{2s-1}$}
	\State \textbf{return} no.
	\EndIf
\EndFor
\If { \textbf{not} $\mathcal{U}_s(\Gamma, b, |A|)$ } \Comment{Function call is in $\rmPi^\P_{2s-1}$}
\State \textbf{return} no.
\EndIf
\State \textbf{return} yes.
\EndProcedure
\end{algorithmic}
\end{algorithm}

The algorithm runs as a \coNP{} machine, returning ``no'' unless every non-deterministic branch returns ``yes''.
Lines 2-8 verify that $A$ is a terminal assembly of bin $b$ (subproblem~1): $A$ is not a terminal assembly if and only if (1) $A$ is not producible (lines 2-4), or (2) another producible assembly $B$ can attach to $A$ (lines 5-8). 
\end{proof}

Every staged system has some number of stages $s \in \mathbb{N}$, but there is no limit to the number of stages a staged system may have. 
Thus the staged UAV problem is not contained in any level of \PH{}, but every instance can be solved by an algorithm that runs at a fixed level ($\rmPi^\P_{2s}$) of the hierarchy.
Since it is a well-known that $\PH{} \subseteq \PSPACE{}$, this gives the desired result: 

\begin{corollary}
The staged UAV problem is in \PSPACE{}. 
\end{corollary}

Next, we move to shape verification:

\begin{definition}[Stage-$s$ unique shape verification (Stage-$s$ USV) problem]
Given a staged system $\Gamma$ with $s$ stages and a shape $S$, is $S$ the unique terminal shape of $\Gamma$?
\end{definition}

\begin{theorem}
\label{thm:stage-s-USV}
The stage-$s$ USV problem is in $\rmPi^\P_{2s}$.
\end{theorem}

\begin{proof}
The stage-$s$ USV problem can be restated as:
\begin{enumerate}
\item is every assembly $B$ with $|B| \leq |S|$ and shape not equal to $S$ not a terminal assembly of any bin in stage $s$?
\item and does every producible assembly of every bin in stage $s-1$ of $\Gamma$ have size at most $|S|$? 
\end{enumerate}

Notice that the subproblems only differ from those of the stage-$s$ UAV problem in that $S$ replaces $A$ and ``equal shape'' replaces ``equals''. 
Thus the algorithm differs from the $\rmPi^\P_{2s}$ algorithm for the stage-$s$ UAV problem on only line~5 (replace ``$A \neq B$'' with ``shape not equal to $S$'') and line~8 (replace $|A|$ with $|S|$).
\end{proof}

As for the UAV problem, since the stage-$s$ USV problem is in \PH{} for each $s \in \mathbb{N}$, the USV problem is in \PSPACE{}.

\begin{corollary}
The staged USV problem is in \PSPACE{}.
\end{corollary}

%% file: open.tex
\section{Open Problems}
\label{sec:open}

The most direct problem left open by this work is closing the gap in the bottom row of Table~\ref{tab:compare} between the $\coNP^{\NP}$-hardness and \PSPACE{} containment of the staged UAV and USV problems.
We believe that the approach of differentiating between satisfying and non-satisying assignments, then checking for the existence of various partial assignments (the $\forall$ portion of \aesat{}) can be generalized to achieve hardness for any number of quantifier alternations, using a number of stages proportional to the number of alternations:

\begin{conjecture}
The staged UAV and USV problems are \PSPACE{}-complete. 
\end{conjecture}

\begin{conjecture}
The stage-$s$ UAV and stage-$s$ USV problems are $\rmPi^p_{\Omega(s)}$-hard.
\end{conjecture}

The UAV and USV problems considered in this work are two variants of the generic challenge of \emph{verification}; considering the same problems limited to temperature-1 systems or with different inputs is also interesting:

\begin{problem}
What are the complexities of the staged UAV and USV problems restricted to temperature-1 systems? 
\end{problem}

\begin{problem}
What is the complexity (in any model) of the following UAV-like problem: given a system $\Gamma$ and an integer $n$, does $\Gamma$ have a unique terminal assembly of size at most $n$?
\end{problem}

Finally, the results and techniques presented here might find use in the study of other problems in staged and two-handed self-assembly, such as tile minimization.
The aTAM USV problem is \coNP-complete, while the \emph{minimum tile set problem} of finding the minimum number of tiles that uniquely assemble into a given shape is $\NP^{\NP}$-complete~\cite{Bryans-2013a}.
We now know that the 2HAM USV problem is $\coNP^{\NP}$-complete (Section~\ref{sec:2HAM}); does the corresponding optimization problem also rise in the hierarchy?
\begin{conjecture}
The 2HAM minimum tile set problem is $\NP^{\NP^{\NP}}$-complete. 
\end{conjecture}